\documentclass[11pt]{article}

\usepackage[margin=0.9in]{geometry}

\usepackage{graphicx}

\usepackage{amssymb}
\usepackage{amsmath}
\usepackage{amsthm}
\usepackage{mathrsfs}
\usepackage[table]{xcolor}
\usepackage{xspace}

\allowdisplaybreaks

\theoremstyle{definition}
\newtheorem{thm}{Theorem}[section]

\newtheorem{proposition}[thm]{Proposition}

\numberwithin{thm}{section}

\newcommand{\mR}{\mathbb{R}}

\newcommand{\mV}{\mathcal{V}}

\newcommand{\fulltoday}{\ifcase\month\or
    January\or February\or March\or April\or May\or June\or
    July\or August\or September\or October\or November\or December\fi\space\number\day,\space\number\year}

\begin{document}
\title{Measuring Concentration of Power in Approval Voting Games}
\author{Takaaki Abe\thanks{Department of Economic Engineering, School of Economics, Kyushu University, 744, Motooka, Nishi-ku, Fukuoka, 819-0395, Japan.
Email: takaakiabe@econ.kyushu-u.ac.jp
}}
\date{}
\maketitle

\begin{abstract}
The ratio of voting power between a permanent member and a non-permanent member of the United Nations Security Council varies substantially across indices: approximately 100 to 1 according to the Shapley-Shubik index, 10 to 1 according to the Banzhaf index, and 2.5 to 1 according to the Deegan-Packel index. These differences make it unclear to what extent power in the UNSC is concentrated in favor of the permanent members.

To address this issue, we propose and characterize a function that measures the level of power concentration in monotonic voting games. 
The proposed measure assigns a single value to each voting game, reflecting the extent to which voting power is unevenly distributed among players. The proposed measure is proportional to the sum of squared Deegan-Packel power indices and can also be interpreted as the degree of overlap among minimal winning coalitions.
An application to the UNSC is also provided.
\end{abstract}


\section{Introduction}\label{SEC_INTRO}

Approval voting is a system of collective decision-making widely used in various organizations, legislatures, and committees.
In the theoretical literature, approval voting is analyzed as an approval voting \textit{game} within the framework of cooperative game theory. Within this framework, a coalition of players is called a winning coalition if its members can approve a proposal without support from any players outside the coalition. An approval voting game is determined by the collection of all winning coalitions. This representation allows us to analyze a wide variety of approval rules in a uniformed framework.

This game-theoretic representation has been applied to real-world decision-making bodies. Prominent examples include the United Nations Security Council, the European Parliament, and shareholder meetings. The main focus of this literature has been the measurement of voting power of individual players. To this end, numerous power indices have been proposed and studied.
Classic examples include the Shapley-Shubik index (Shapley and Shubik, 1954), the Banzhaf index (Banzhaf, 1965), the Johnston index (Johnston, 1978), the Deegan-Packel index (Deegan and Packel, 1978), and the Holler index (Holler, 1982). 
This line of research has been remarkably successful in advancing our understanding of voting systems.

Despite these achievements, the following question has received relatively little attention: How should \textit{concentration} of voting power in an approval voting game be measured? 
To illustrate this issue, consider the United Nations Security Council (UNSC). The UNSC consists of five permanent members and ten non-permanent members. Under its current decision rule, a proposal is approved if it receives affirmative votes from all five permanent members and at least four non-permanent members.
When this rule is evaluated using the Shapley-Shubik power index, a permanent member is estimated to have about one hundred times the voting power of a non-permanent member. This large disparity offers the ``impression'' that voting power in the UNSC is highly concentrated. In contrast, when the Banzhaf index is applied, the corresponding ratio is approximately 10 to 1. Using the Deegan-Packel index further reduces the ratio to about 2.5 to 1, giving a much less concentrated impression of power distribution.
Such discrepancies may give critics of a voting rule an incentive to emphasize power indices that make the degree of concentration appear large, while giving supporters an incentive to use alternative indices that make concentration appear small. This observation suggests the need for a justified measure that evaluates the degree of power concentration in a unified manner.

The aim of this paper is to propose a measure of concentration of power for approval voting games. The measure is derived from a set of axioms.
The measure applies to voting games with different player sets, allowing comparisons between games with many players and those with fewer players.

A novelty of this paper lies in the usage of the degree of overlap among sets.
Monotonic voting games are fully described by their sets of minimal winning coalitions. Our key idea is that the extent of overlap among minimal winning coalitions provides a natural basis for measuring power concentration.
When the same players appear repeatedly across multiple minimal winning coalitions, power is concentrated in their hands. In contrast, when minimal winning coalitions involve a wide range of players, power is dispersed.

The remainder of the paper is organized as follows. 
Section \ref{SEC_PREL} introduces basic definitions and notation.
Section \ref{SEC_MEASURE} proposes axioms for measuring the concentration of voting power. 
Section \ref{SEC_CHAR} provides an axiomatic characterization of the proposed concentration measure. 
Section \ref{SEC_ALT} examines alternative measures. 
Section \ref{SEC_APL} applies the proposed measure to the UNSC and its variants. 
Section \ref{SEC_CONC} provides concluding remarks.

\section{Preliminaries}\label{SEC_PREL}
Let $N \subsetneq \mathbb{N}$ be a nonempty finite set of players, and let $2^{N}$ denote the set of all subsets of $N$. A subset of $N$ is called a \textit{coalition} of players. For a set $X$, let $|X|$ denote the number of elements of $X$. An \textit{approval voting game} is a pair $(N, W)$, where $W \subseteq 2^{N}$ ($\emptyset \notin W$ and $N\in W$) denotes the set of winning coalitions. For simplicity, we refer to an approval voting game simply as a ``voting game.''

A voting game $(N, W)$ is \textit{monotonic} if for every $S, T \in 2^{N}$ with $S \subseteq T$, $S \in W$ implies $T \in W$. A winning coalition $S \in W$ is \textit{minimal} if no proper subset of $S$ is winning. For every $W \subseteq 2^{N}$, let $M(W)$ denote the set of all minimal winning coalitions in $W$. Let $\mV^{N}$ be the set of all monotonic voting games with player set $N$.

Let $(N, W) \in \mathcal{V}^{N}$ be a voting game. Define the \textit{characteristic function} $v_W : 2^{N} \rightarrow \{0,1\}$ associated with $W$ by
\[
v_W(S) =
\begin{cases}
1 & \text{if } S \in W,\\
0 & \text{if } S \notin W.
\end{cases}
\]
For every nonempty $T\subset N$, a voting game is called a $T$-unanimity game $u_T\in \mV^N$ if
\[
u_T(S) =
\begin{cases}
1 & \text{if } S \supseteq T,\\
0 & \text{otherwise.}
\end{cases}
\]

A player $i \in N$ is a \textit{null player} in $W$ if for every $S \subseteq N \setminus \{i\}$, $v_W(S \cup \{i\}) - v_W(S) = 0$.
Players $i$ and $j$ are \textit{symmetric} in $W$ if for every $S \subseteq N \setminus \{i,j\}$, $v_W(S \cup \{i\}) =v_W(S \cup \{j\}) $. A voting game $(N,W)$ is \textit{symmetric} if all players in $N$ are symmetric in $W$.

\section{Concentration measures and axioms}\label{SEC_MEASURE}

In this section, we formulate a function that measures the concentration of power in a voting game and propose axioms for concentration measures. Such a measure can be defined as a function that assigns a real number to each voting game. Let
$\mV= \{ (N, W) \in \mV^{N} \mid N \subsetneq \mathbb{N},\ 0< |N| < \infty \}$
denote the set of all monotonic voting games with finite nonempty player sets. A \textit{concentration measure} is then defined as a function
$\mu:\mV \rightarrow \mR$.
In words, the value $\mu(N,W)$ captures the extent to which voting power is concentrated among players in the voting game $(N,W)$. A higher value of $\mu(N,W)$ indicates that power is concentrated in the hands of a few players, whereas a lower value reflects a dispersed distribution of power across many players. We now introduce axioms for concentration measures.

\subsection{Symmetry}

We begin by considering symmetric voting games. Intuitively, symmetric games should exhibit a low degree of power concentration.
More precisely, a symmetric voting game should not display a greater level of power concentration than any other, possibly asymmetric, game defined on the same player set. Furthermore, any two symmetric games on the same player set should be regarded as equivalent in terms of power concentration. These considerations motivate the following requirement.

\begin{itemize}
\item[] Strong Symmetric Game Property (\textbf{SSYM}).  
For every nonempty $N \subseteq \mathbb{N}$ and every $(N,W)$, $(N,W')$ $\in \mathcal{V}^{N}$, if $(N,W)$ is symmetric, then $\mu(N,W) \leq \mu(N,W')$.
\end{itemize}

This condition can be weakened by extracting the equality requirement among symmetric games.

\begin{itemize}
\item[] Symmetric Game Property (\textbf{SYM}).  
For every nonempty $N \subseteq \mathbb{N}$ and every $(N,W),(N,W') \in \mathcal{V}^{N}$, if $(N,W)$ and $(N,W')$ are symmetric, then $\mu(N,W) = \mu(N,W')$.
\end{itemize}

For a given player set $N$, symmetric voting games in $\mathcal{V}^{N}$ can be specified by two natural numbers $(n,k)$ with $n\geq k$, where $k$ is the threshold number of affirmative votes required for approval. In other words, $(N,W)$ is a symmetric game if and only if $W = \{ S \subseteq N \mid |S| \geq k \}$ for some natural number $k$.
The next axiom weakens the symmetric game property by restricting attention to two extreme symmetric cases, rather than requiring equality across all symmetric games. These cases correspond to the thresholds $k = n$ and $k = 1$. The former yields the all-player unanimity game $(N,u_N)$. The latter corresponds to the rule under which every nonempty coalition is winning.
The following axiom requires a concentration measure to assign the same value to these two extreme symmetric cases.

\begin{itemize}
\item[] Weak Symmetric Game Property (\textbf{WSYM}).  
For every nonempty $N \subseteq \mathbb{N}$. $\mu(N,u_N) = \mu(N,W^*_{N})$.
\end{itemize}

\subsection{Replication}

We next introduce the notion of \textit{replicating} players, which captures the idea that power concentration should decrease when each player is replaced by multiple identical copies. Replication can be interpreted as subdividing existing political or voting entities into smaller, identical units, thereby ``diluting'' individual influence while preserving the overall structure of the voting game.

To illustrate this idea, consider the following voting game $(N,W)$ with player set $N=\{1,2,3\}$ and minimal winning coalitions $M(W)=\{\{1,2\}, \{1,3\}\}$.
We replicate each player into two identical copies: for every $i\in \{1,2,3\}$, player $i$ is replaced by the players $i$ and $i'$. The resulting player set is therefore $N'=\{1,1',2,2',3,3'\}$. Replacing each original player $i$ with its replicas $i$ and $i'$ yields replicated minimal winning coalitions $M(W')=\{ \{1,1',2,2'\}, \{1,1',3,3'\} \}$.\footnote{Since players are represented by natural numbers, labels $1', 2', 3'$ may be relabeled as, for example, $4,5,6$.}
We interpret the replicated voting game $(N',W')$ as being twice as ``dilute'' as the original game $(N,W)$ in terms of concentration. In other words, although the structure of winning coalitions of the game is preserved, each original position is now shared by two identical players. 

We now present a general definition of replication. Consider two nonempty player sets $N,N' \subsetneq \mathbb{N}$ and voting games $W\in \mV^N$ and $W'\in \mV^{N'}$.
For a natural number $t$, voting game $(N', W')$ is a $t$-\textit{replica} of $(N, W)$ if there is $\rho: N\twoheadrightarrow N'$ such that for all distinct $i,j\in N$, $\rho(i)\cap \rho(j)=\emptyset$, $|\rho(i)|= |\rho(j)|=t$, and $\cup_{i\in N}\rho(i)=N'$.
As described in the example above, the interpretation is that the $t$-replicated game $(N',W')$ is $t$ times as dilute as the original game $(N,W)$ in terms of concentration. 
This intuition leads to the following axiom.

\begin{itemize}
\item[] Replication (\textbf{REP}).
For every $t\in \mathbb{N}$ and every $(N,W),(N',W')\in \mV$, if $(N',W')$ is a $t$-replica of $(N,W)$, then $\mu(N',W')=\frac{1}{t}\mu(N,W)$.
\end{itemize}

Replication is closely related to relabeling players. When $t=1$, the mapping $\rho$ simply associates each player in $N$ with a player in $N'$, and the two player sets may or may not coincide. In this sense, the replication axiom incorporates a form of player anonymity: measure $\mu$ depends on the structure of winning coalitions but not on the labels of players.

We now introduce a weak version of the replication axiom, which we call \textit{Duplication}. This axiom restricts attention to the special case in which (i) $t=2$ and (ii) $|N|=1$ (and hence $|N'|=2$). When $|N|=1$, say $N=\{i\}$, the one-person unanimity game $u_{\{i\}}$ is the only game in $\mV^N$. Its $2$-replica must therefore be a two-person unanimity game $u_{\{i,j\}}$ for some $j\in\mathbb{N}$.
The following axiom imposes the replication requirement only on this elementary case.

\begin{itemize}
\item[] Duplication (\textbf{DUP}).
For all distinct $i,j\in \mathbb{N}$, $\mu(\{i,j\},u_{\{i,j\}})=\frac{1}{2}\mu(\{i\},u_{\{i\}})$.
\end{itemize}

Therefore, it immediately follows that \textbf{REP} $\Rightarrow$ \textbf{DUP}.

\subsection{Minimal winning coalitions}
In the class of monotonic voting games, for a given player set $N$, the collection $M(W)$ fully determines the voting game $W$. This observation allows us to describe a voting game by its set of minimal winning coalitions.
The following requirement states that if two voting games, possibly defined on different player sets, share exactly the same collection of minimal winning coalitions, then they should be regarded as identical in terms of concentration.

\begin{itemize}
\item[] Minimal Winning Coalition Property (\textbf{MWC}).
For every $(N,W),(N',W') \in \mV$ with $M(W)=M(W')$, $\mu(N,W)=\mu(N',W')$.
\end{itemize}

Note that, unlike the player anonymity, this axiom requires the equivalence between the constituent members of $M(W)$ and those of $M(W')$.
This requirement is related to the null-player out property (Derks and Haller, 1999) used in the axiomatic analysis of the Shapley value.\footnote{See, for example, B\'{e}al et al. (2016) and Kongo (2024) for recent related studies.} Indeed, when $(N,W)$ and $(N',W')$ satisfy $M(W)=M(W')$, the only difference between the two games lies in null players. If $N \neq N'$, then the difference $N'\setminus N$ or $N\setminus N'$ consists entirely of null players, since any player who does not belong to any minimal winning coalition is a null player in a monotonic voting game.
Therefore, this requirement equivalently states that the presence or absence of null players is irrelevant for the evaluation of power concentration. In this sense, concentration is about the structure of decisive coalitions, and  the axiom rules out measures that are sensitive to the number or the labeling of null players.

\subsection{Average}

Additivity does not naturally fit the domain of voting games.\footnote{Dubey (1975) axiomatically studied the Shapley value within the domain of voting games and proposed an axiom called the Transfer axiom instead of Additivity. Einy and Haimanko (2011) employed this axiom and provided an axiomatic characterization of the Shapley-Shubik power index (discussed in Section \ref{SEC_ALT}).} If we take two voting games $W$ and $W'$ and represent them by their characteristic functions $v_W$ and $v_{W'}$, then any attempt to ``add'' the two games results in the function $v_W + v_{W'}$, which no longer belongs to the class of voting games.

Moreover, additivity would be conceptually inappropriate for measuring power concentration. If a measure is additive, then adding a game with concentration level 1 to a game with concentration level 1 would yield a game with concentration level 2. Such an operation lacks a meaningful interpretation especially in terms of concentration. For these reasons, an alternative aggregation principle is required.

Since the set of all monotonic voting games is a lattice, every monotonic voting game $W$ can be represented as the join of the unanimity games $u_S$ corresponding to its minimal winning coalitions $S\in M(W)$. In this sense, minimal winning coalitions serve as  ``basis'' of voting games. It is therefore natural to postulate that the concentration of a voting game should be derived from the concentration levels of these constituent unanimity games.
The most straightforward way to formalize this intuition can be described as a simple average:
\[
\mu(N,W)=\frac{1}{|M(W)|} \sum_{S\in M(W)} \mu(N,u_S).
\]
However, this naive average leads to counterintuitive conclusions. To illustrate the problem, consider $N=\{1,2,3,4\}$ and two voting games with $M(W)=\{\{1,2\},\{2,3\}\}$ and $M(W')=\{\{1,2\},\{3,4\}\}$.
The averaging formula above implies
\begin{eqnarray*}
\mu(N,W) &=& \frac{1}{2}( \mu(N,u_{\{1,2\}}) + \mu(N,u_{\{2,3\}}) ),\\
\mu(N,W') &=& \frac{1}{2}( \mu(N,u_{\{1,2\}}) + \mu(N,u_{\{3,4\}}) ).
\end{eqnarray*}
If one admits $\mu(N,u_{\{2,3\}}) = \mu(N,u_{\{3,4\}})$, then it follows that $\mu(N,W)=\mu(N,W')$.
However, this is implausible: in game $W$, player 2 belongs to all winning coalitions, which suggests a higher concentration of power than in $W'$, where no player occupies such a central position. This stems from the fact that it treats minimal winning coalitions as independent components and ignores the extent to which they \textit{overlap}. 

A standard way to quantify the degree of overlap between two nonempty coalitions $S$ and $T$ is the Dice index (Dice, 1945),
\[
\delta(S,T)=2\cdot \frac{|S\cap T|}{|S|+|T|}.
\]
The Dice index satisfies $0 \leq \delta(S,T) \leq 1$, attaining $0$ when $S$ and $T$ are disjoint and $1$ when $S=T$. It therefore provides a natural measure of overlap between coalitions.
Incorporating the Dice index as a weight, we obtain the following averaging formula:
\[
\mu(N,W)=\frac{1}{|M(W)|^2} \sum_{(S,T)\in M(W)^2} \delta(S,T) \cdot \frac{\mu(N,u_S)+\mu(N,u_T)}{2}.
\]
The interpretation is straightforward. Two minimal winning coalitions, $S$ and $T$, are randomly selected, their respective unanimity games are averaged, and this average is weighted by the degree, $\delta(S,T)$, to which the coalitions overlap. 
This idea naturally extends to any number $k$ of selected minimal winning coalitions.
For a natural number $k$, define
\[
\delta(S_1,...,S_k)=k\cdot \frac{ |S_1\cap ... \cap S_k|}{|S_1|+...+|S_k|}
\]
and consider
\[
\mu(N,W)=\frac{1}{|M(W)|^k} \sum_{(S_1,...,S_k)\in M(W)^k} \delta(S_1,...,S_k) \cdot \frac{\mu(N,u_{S_1})+...+\mu(N,u_{S_k})}{k}.
\]
When $k=1$, this formula reduces to the naive averaging formula introduced above.\footnote{When $k=1$, it holds that $\delta(S_1)=1 \cdot \frac{|S_1|}{|S_1|} =1$.}
We therefore formulate the Average axiom as follows.

\begin{itemize}
\item[] Average (\textbf{AVE}).
There is a natural number $k$ such that for every $(N,W)\in \mV$,
\[
\mu(N,W)=\frac{1}{|M(W)|^k} \sum_{(S_1,...,S_k)\in M(W)^k} \delta(S_1,...,S_k) \cdot \frac{\mu(N,u_{S_1})+...+\mu(N,u_{S_k})}{k}.
\]
\end{itemize}

For example, the measure $\mu(N,W)=1$ for every $(N,W) \in \mV$ satisfies \textbf{AVE} with $k=1$. Moreover, $\mu(N,W)=\frac{1}{|N|}$ for every $(N,W) \in \mV$ also satisfies \textbf{AVE} with $k=1$.

\section{Axiomatic characterization}\label{SEC_CHAR}
\subsection{Characterization}
The following proposition shows that the axioms \textbf{WSYM}, \textbf{DUP}, \textbf{MWC}, and \textbf{AVE} uniquely determine a particular measure.

\begin{proposition}\label{PROP_HDP}
A measure $\mu:\mV\rightarrow \mR$ satisfies \textbf{WSYM}, \textbf{DUP}, \textbf{MWC}, and \textbf{AVE} if and only if there is $a\in \mR$ such that for every $(N,W)\in \mV$,
\[
\mu(N,W)= a\cdot \sum_{i\in N} DP_i(N,W)^2,
\]
where $DP_i(N,W)$ denotes the Deegan-Packel power index, $DP_i(N,W)=\frac{1}{|M(W)|} \sum_{S:i\in S\in M(W)} \frac{1}{|S|}$.
\end{proposition}
\begin{proof}
We first prove the uniqueness part.
Let $\mu:\mV\rightarrow \mR$ be a function that satisfies \textbf{WSYM}, \textbf{DUP}, \textbf{MWC}, and \textbf{AVE}.
Consider player $1\in \mathbb{N}$.
Let $a:=\mu(\{1\}, u_{\{1\}})$.
For every $j\in \mathbb{N}$ with $j\neq 1$,
\begin{equation}
\mu(\{1,j\}, u_{\{1,j\}}) \overset{\text{\textbf{DUP}}}{=} \frac{1}{2} \mu(\{1\}, u_{\{1\}}) = \frac{1}{2} a. \label{eq_0309_1439}
\end{equation}
Moreover, we have $\mu(\{1,j\}, u_{\{1,j\}}) \overset{\text{\textbf{DUP}}}{=} \frac{1}{2} \mu(\{j\}, u_{\{j\}})$. Hence, from (\ref{eq_0309_1439}), it follows that for every $j\in \mathbb{N}$,
\begin{equation}
\mu(\{j\}, u_{\{j\}})=a. \label{eq_0309_1443}
\end{equation}

\textbf{AVE} implies that there is a natural number $k$ such that for every $(N,W)\in \mV$,
\begin{eqnarray}
\mu(N,W)
&=&\frac{1}{|M(W)|^k} \sum_{(S_1,...,S_k)\in M(W)^k} \delta(S_1,...,S_k) \cdot \frac{\mu(N,u_{S_1})+...+\mu(N,u_{S_k})}{k} \nonumber\\
&=&\frac{1}{|M(W)|^k} \sum_{(S_1,...,S_k)\in M(W)^k} \frac{ |S_1\cap ... \cap S_k|}{|S_1|+...+|S_k|} (\mu(N,u_{S_1})+...+\mu(N,u_{S_k})). \label{eq_0309_1459}
\end{eqnarray}
Consider monotonic voting game $(\{1,2\}, W^*_{\{1,2\}})$, where $W^*_{\{1,2\}}$ is defined as $M(W^*_{\{1,2\}})=\{\{1\},\{2\}\}$.
We have
\begin{eqnarray}
\mu(\{1,2\}, W^*_{\{1,2\}})
&\overset{(\ref{eq_0309_1459})}{=}&\frac{1}{2^k} \left( \frac{1}{k\cdot 1}\cdot k \cdot \mu(\{1,2\}, u_{\{1\}}) + \frac{1}{k\cdot 1}\cdot k \cdot \mu(\{1,2\}, u_{\{2\}})  \right) \nonumber\\
&=&\frac{1}{2^k} \left( \mu(\{1,2\}, u_{\{1\}}) + \mu(\{1,2\}, u_{\{2\}})  \right) \nonumber\\
&\overset{\text{\textbf{MWC}}}{=}& \frac{1}{2^k} \left( \mu(\{1\}, u_{\{1\}}) + \mu(\{2\}, u_{\{2\}})  \right) \nonumber\\
&\overset{(\ref{eq_0309_1443})}{=}& \frac{1}{2^{k-1}} \cdot a. \label{eq_0309_1500}
\end{eqnarray}
Hence, we obtain
\[
\frac{1}{2^{k-1}} \cdot a \overset{(\ref{eq_0309_1500})}{=} \mu(\{1,2\}, W^*_{\{1,2\}}) \overset{\text{\textbf{WSYM}}}{=} \mu(\{1,2\}, u_{\{1,2\}}) \overset{(\ref{eq_0309_1439})}{=} \frac{1}{2}\cdot a.
\]
Therefore, we have either $a=0$ or ($a\neq 0$ and $k=2$).

In addition, let $N\subsetneq \mathbb{N}$. For every nonempty $S\subseteq N$, let $(S,W^*_{S})$ be the monotonic voting game defined as $M(W^*_{S})=\{\{i\}\mid i\in S\}$. We have
\begin{eqnarray}
\mu(N,u_S)
&\overset{\text{\textbf{MWC}}}{=}&\mu(S,u_S)\nonumber\\
&\overset{\text{\textbf{WSYM}}}{=}&\mu(S, W^*_{S})\nonumber\\
&\overset{(\ref{eq_0309_1459})}{=}& \frac{1}{|S|^k} \sum_{j\in S} \frac{1}{k} \cdot k \cdot \mu(S,u_{\{j\}})\nonumber\\
&\overset{\text{\textbf{MWC}}}{=}& \frac{1}{|S|^k} \sum_{j\in S} \frac{1}{k} \cdot k \cdot \mu(\{j\},u_{\{j\}})\nonumber\\
&\overset{(\ref{eq_0309_1443})}{=}& \frac{1}{|S|^k} \sum_{j\in S} \frac{1}{k} \cdot k \cdot a\nonumber\\
&=& \frac{a}{|S|^{k-1}}. \label{eq_0311_0931}
\end{eqnarray}

We first consider the case with $a=0$.
For every nonempty $S\subseteq N$, by $a=0$,
\begin{equation}
\mu(N,u_S) \overset{(\ref{eq_0311_0931})}{=} 0.\label{eq_0311_0943}
\end{equation}
Therefore, for every $(N,W)\in \mV$, we obtain $\mu(N,W)\overset{(\ref{eq_0309_1459}), (\ref{eq_0311_0943})}{=}0$. This is a desired functional form $\mu(N,W)= a\cdot \sum_{i\in N} DP_i(N,W)^2$ with $a=0$.

We now consider the case with $a\neq 0$ and $k=2$.
For every nonempty $S\subseteq N$, by $a\neq 0$ and $k=2$,
\begin{equation}
\mu(N,u_S) \overset{(\ref{eq_0311_0931})}{=} \frac{a}{|S|}.\label{eq_0311_0946}
\end{equation}
Therefore, for every $(N,W)\in \mV$, it follows from $k=2$ that 
\begin{eqnarray}
\mu(N,W)
&\overset{(\ref{eq_0309_1459})}{=}&\frac{1}{|M(W)|^2} \sum_{(S,T)\in M(W)^2} \frac{ |S\cap T|}{|S|+|T|} (\mu(N,u_{S})+\mu(N,u_{T}))\nonumber\\
&\overset{(\ref{eq_0311_0946})}{=}&\frac{1}{|M(W)|^2} \sum_{(S,T)\in M(W)^2} \frac{ |S\cap T|}{|S|+|T|} \left( \frac{a}{|S|}+\frac{a}{|T|} \right)\nonumber\\
&=&\frac{a}{|M(W)|^2} \sum_{(S,T)\in M(W)^2} \frac{ |S\cap T|}{|S|\cdot|T|}.\label{eq_0311_1021}
\end{eqnarray}
Moreover, we have
\begin{eqnarray*}
(\ref{eq_0311_1021})
&=&\frac{a}{|M(W)|^2} \sum_{(S,T)\in M(W)^2} \sum_{i\in S \cap T} \frac{1}{|S|\cdot|T|}\\
&=&\frac{a}{|M(W)|^2} \sum_{i\in N} \ \sum_{\substack{(S,T)\in M(W)^2 \\ i\in S \cap T}} \frac{1}{|S|\cdot|T|}\\
&=&\frac{a}{|M(W)|^2} \sum_{i\in N} \ \sum_{\substack{(S,T)\in M(W)^2 \\ i\in S, i\in T}} \frac{1}{|S|} \cdot \frac{1}{|T|}\\
&=&\frac{a}{|M(W)|^2} \sum_{i\in N} \left( \sum_{\substack{S\in M(W) \\ i\in S}} \frac{1}{|S|}\right)^2\\
&=&a \sum_{i\in N} \left( \frac{1}{|M(W)|} \sum_{\substack{S\in M(W) \\ i\in S}} \frac{1}{|S|}\right)^2\\
&=&a \sum_{i\in N} DP_i(N,W)^2.
\end{eqnarray*}
This completes the uniqueness part.

We provide the sufficiency part for completeness.
Let $a\in \mR$. Let $\mu(N,W)= a\cdot \sum_{i\in N} DP_i(N,W)^2$ for every $(N,W)\in \mV$.

\noindent \textbf{WSYM}: For every $N\subsetneq \mathbb{N}$, $DP(N,u_N)=(\frac{1}{n},...,\frac{1}{n})=DP(N,W^*_{N})$. Hence, we have $\mu(N,u_N)=\mu(N,W^*_{N})$.

\noindent \textbf{DUP}: For all distinct $i,j\in \mathbb{N}$, $\mu(\{i,j\},u_{\{i,j\}})= a \cdot (\frac{1}{2}^2+\frac{1}{2}^2) = \frac{a}{2} =\frac{1}{2}\mu(\{i\},u_{\{i\}})$.

\noindent \textbf{MWC}: For every $(N,W),(N',W') \in \mV$ with $M(W)=M(W')$, it holds that $\sum_{i\in N}DP_i(N,W)^2=\sum_{i\in N'}DP_i(N',W')^2$. Hence, we have $\mu(N,W)=\mu(N',W')$.

\noindent \textbf{AVE}: For every $(N,W) \in \mV$, as described in (\ref{eq_0311_1021}), we have $\mu(N,W)=\frac{a}{|M(W)|^2} \sum_{(S,T)\in M(W)^2} \frac{ |S\cap T|}{|S|\cdot|T|}$.
Hence, for every nonempty $S\subseteq N$, $\mu(N,u_S)=\frac{a}{|S|}$.
Therefore, it holds that
\[\mu(N,W)=\frac{1}{|M(W)|^2} \sum_{(S,T)\in M(W)^2} \frac{ |S\cap T|}{|S|+|T|} (\mu(N,u_{S})+\mu(N,u_{T})).\]
Hence, the function $\mu$ satisfies \textbf{AVE} with $k=2$.
\end{proof}

Proposition \ref{PROP_HDP} states that the characterized measure $\mu$ is the sum of squared Deegan-Packel power shares across players. However, as shown in the proof (see equation (\ref{eq_0311_1021}) with $a=1$), we can alternatively express this measure as
$\frac{1}{|M(W)|^2} \sum_{(S,T)\in M(W)^2} \frac{ |S\cap T|}{|S|\cdot|T|}$.
This expression can be understood as measuring a certain degree of overlap among minimal winning coalitions.
In addition, this measure can be viewed as an application of the Herfindahl-Hirschman (HH) index to the Deegan-Packel power index.
The HH index is a classical measure of concentration that was axiomatically characterized by Hall and Tideman (1967) in the context of industrial organization. In that literature, the HH index is defined as the sum of squared market shares of firms and is widely used to assess the degree of concentration in an industry. 
Higher values of the HH index indicate that market activity is concentrated among a small number of firms, whereas lower values correspond to a more competitive market structure.
Our measure inherits this interpretation in the domain of voting games.

There is a clear technical difference between the axiomatic characterization of the HH index and that of a power concentration measure for voting games. The key distinction lies in continuity. 
The HH index is defined on vectors of market shares. Hence, its domain is a subset of $n$-dimensional real space. Therefore, it is natural to require the original HH index to satisfy continuity on the domain.
By contrast, the domain of a power concentration measure is the class of monotonic voting games. Each voting game is represented by a set of winning coalitions.\footnote{Even when a voting game is represented by a characteristic function, it corresponds to a binary vector with $2^n$ elements taking only values in $\{0,1\}$.} As a consequence, there is no straightforward analogue of continuity on this domain.

One might conjecture that the focus on minimal winning coalitions, as embodied in axiom \textbf{MWC}, leads to the use of the DP index in Proposition \ref{PROP_HDP}. However, this conjecture is not true. 
In fact, even if one applies the HH index to the Shapley-Shubik power index, the resulting concentration measure satisfies \textbf{MWC}. This point is examined in detail in Section \ref{SEC_ALT}.

\subsection{Axiomatic independence of Proposition \ref{PROP_HDP}}
The following functions show that the axioms used in Proposition \ref{PROP_HDP} are independent.
\begin{itemize}
\item Violating \textbf{WSYM}: Let $f^1(N,W)=\frac{1}{|M(W)|} \sum_{S\in M(W)} \frac{1}{|S|}$ for every $(N,W) \in \mV$. Function $f^1$ satisfies \textbf{DUP}, \textbf{MWC}, and \textbf{AVE} (with $k=1$). Function $f^1$ violates \textbf{WSYM} as $f^1(\{1,2\}, u_{\{1,2\}}) = \frac{1}{2} \neq 1 = \frac{1}{2}(1+1) = f^1(\{1,2\}, W^*_{\{1,2\}})$, where $M(W^*_{\{1,2\}})=\{\{1\},\{2\}\}$.

\item Violating \textbf{DUP}: Let $f^2(N,W)=1$ for every $(N,W) \in \mV$. Function $f^2$ satisfies \textbf{WSYM}, \textbf{MWC}, and \textbf{AVE} (with $k=1$). Function $f^2$ violates \textbf{DUP} as $f^2(\{1,2\}, u_{\{1,2\}}) = 1 \neq \frac{1}{2} = \frac{1}{2} f^2(\{1\}, u_{\{1\}})$.

\item Violating \textbf{MWC}: Let $f^3(N,W)=\frac{1}{|N|}$ for every $(N,W) \in \mV$. Function $f^3$ satisfies \textbf{WSYM}, \textbf{DUP}, and \textbf{AVE} (with $k=1$). Function $f^3$ violates \textbf{MWC} as $f^3(\{1,2\}, u_{\{1\}}) = \frac{1}{2} \neq 1 = f^3(\{1\}, u_{\{1\}})$.

\item Violating \textbf{AVE}: Let $f^4(N,W)=\frac{1}{|\cup_{S\in M(W)}S|}$ for every $(N,W) \in \mV$. Function $f^4$ satisfies \textbf{WSYM}, \textbf{DUP}, and \textbf{MWC}. We show that $f^4$ violates \textbf{AVE}.
Assume that $f^4$ satisfies \textbf{AVE}. There is $k\in \mathbb{N}$ such that for every $(N,W) \in \mV$,
\begin{equation}
f^4(N,W)=\frac{1}{|M(W)|^k} \sum_{(S_1,...,S_k)\in M(W)^k} \delta(S_1,...,S_k) \cdot \frac{f^4(N,u_{S_1})+...+f^4(N,u_{S_k})}{k}. \label{eq_0311_1451}
\end{equation}
Therefore, we must prove that for every $k\in \mathbb{N}$, there is $(N,W) \in \mV$ that violates this equality. Let $k\in \mathbb{N}$. Consider $(\{1,2,3\}, \bar{W})$, where $M(\bar{W})=\{\{1,2\},\{1,3\}\}$. By the definition of $f^4$, we have $f^4(\{1,2,3\}, \bar{W})=\frac{1}{|\{1,2\}\cup \{1,3\}|}=\frac{1}{3}$, $f^4(\{1,2,3\}, u_{\{1,2\}})=\frac{1}{2}$, and $f^4(\{1,2,3\}, u_{\{1,3\}})=\frac{1}{2}$.
Hence, it holds that 
\[
\frac{1}{3}=f^4(\{1,2,3\},\bar{W})\overset{(\ref{eq_0311_1451})}{=}\frac{1}{2^k} \cdot \frac{1}{2k} \cdot (2^k +2) \cdot \frac{k}{2}=\frac{1}{4}+\frac{1}{2^{k+1}}.
\]
This is equivalent to $\frac{1}{12}=\frac{1}{2^{k+1}}$. However, this does not hold as $12=2^2\cdot 3$ and $k$ is a natural number. This is a contradiction.
\end{itemize}

\section{Alternative measure possibilities}\label{SEC_ALT}

\subsection{Combination of the HH index and the Shapley-Shubik power index}
The measure proposed in Proposition \ref{PROP_HDP} can be viewed as a composite of the HH index and the DP index.
One might consider replacing the DP index with the Shapley-Shubik power index (the SS index).
The SS index is defined as follows: for every $(N,W)\in \mV$ and $i\in N$,
\[
SS_{i}(N,W) = \sum_{S \subseteq N \setminus \{i\}} \frac{|S|!\, (|N|-|S|-1)!}{|N|!} \left( v_{W}(S \cup \{i\}) - v_{W}(S) \right).
\]
The following result shows that the measure obtained from combining the HH index and the SS index depends only on minimal winning coalitions.
For notational simplicity, below we write $\cup_L:= \cup_{S\in L}S$ for every nonempty $L\subseteq M(W)$.

\begin{proposition}\label{PROP_HSS}
For every $(N,W)\in \mV$, it holds that
\[
\sum_{i\in N}SS_i(N,W)^2 = \sum_{\substack{\emptyset \neq L \subseteq M(W) \\ \emptyset \neq K \subseteq M(W)}} (-1)^{|L|+|K|} \frac{|\cup_L \cap \cup_K|}{|\cup_L|\cdot |\cup_K|}.
\]
\end{proposition}
\begin{proof}
Einy and Haimanko (2011)\footnote{See Equation (8) in the proof of Lemma1.} show that for every $(N,W)\in \mV$, it holds that
\[
SS_i(N,W)=\sum_{\emptyset \neq L \subseteq M(W)} (-1)^{|L|+1} \cdot \eta_i(\cup_L),
\]
where for every nonempty $S\subseteq N$,
\[
\eta_i(S)=
\begin{cases}
\frac{1}{|S|} & \text{ if $i\in S$}, \\
0 & \text{ otherwise}. \\
\end{cases}
\]
Hence, for every $(N,W)\in \mV$, we have the following:
\begin{eqnarray*}
\sum_{i\in N}SS_i(N,W)^2 
&=& \sum_{i\in N} \left( \sum_{\emptyset \neq L \subseteq M(W)} (-1)^{|L|+1} \cdot \eta_i(\cup_L) \right)^2 \\
&=& \sum_{i\in N} \left( \sum_{\emptyset \neq L \subseteq M(W)} (-1)^{|L|+1} \cdot \eta_i(\cup_L) \right) \left( \sum_{\emptyset \neq K \subseteq M(W)} (-1)^{|K|+1} \cdot \eta_i(\cup_K) \right) \\
&=& \sum_{i\in N} \sum_{\substack{\emptyset \neq L \subseteq M(W) \\ \emptyset \neq K \subseteq M(W)}} (-1)^{|L|+|K|} \cdot \eta_i(\cup_L) \cdot \eta_i(\cup_K) \\
&=& \sum_{\substack{\emptyset \neq L \subseteq M(W) \\ \emptyset \neq K \subseteq M(W)}} (-1)^{|L|+|K|} \sum_{i\in N}   \eta_i(\cup_L) \cdot \eta_i(\cup_K) \\
&=& \sum_{\substack{\emptyset \neq L \subseteq M(W) \\ \emptyset \neq K \subseteq M(W)}} (-1)^{|L|+|K|}  \frac{|\cup_L \cap \cup_K|}{|\cup_L|\cdot |\cup_K|}.
\end{eqnarray*}
This completes the proof.
\end{proof}
This concentration measure, $\sum_{i\in N}SS_i(N,W)^2$, satisfies \textbf{WSYM}, \textbf{DUP}, and \textbf{MWC}, while it violates \textbf{AVE}.

\subsection{Sum of Deegan-Packel shares to the power of $k$}
Proposition \ref{PROP_HDP} shows that the axioms characterize the sum of \textit{squared} DP shares across players. A natural question is whether this result extends beyond the quadratic case. In particular, one may ask which axioms are violated if, instead of squares, we consider the sum of DP shares raised to an arbitrary power $k$.
For every natural number $k$, define
\[
\mu^k_{\text{DP}}(N,W)= \sum_{i\in N} DP_i(N,W)^k.
\]

Table \ref{TBL_HDP} summarizes which axioms are satisfied by $\mu^k_{\text{DP}}(N,W)$ for different values of $k$.
If $k=1$ or $k=2$, the measure $\mu^k_{\text{DP}}(N,W)$ satisfies \textbf{AVE} with $k'=1$ and $k'=2$, respectively, where $k'$ denotes the number $k$ used in the definition of \textbf{AVE}.\footnote{Note that $\mu^{k=1}_{\text{DP}}$ assigns 1 to every $(N,W)$.}
However, if $k\geq 3$, \text{AVE} is violated.
Specifically, for every natural number $k'$, there exists a voting game for which the equation of \textbf{AVE} fails to hold.\footnote{For example, if $k'$ is different from the $k$ used in $\mu^k_{\text{DP}}$, then the equation of \textbf{AVE} does not hold for the voting game $W$ with $M(W)=\{\{1\},\{2\}\}$. If $k'=k$, the equation fails for the voting game $W$ with $M(W)=\{\{1,2\},\{1,3,4\}\}$.}
This observation highlights a distinction between the quadratic case and higher-order power aggregations.

\begin{table}[htbp]
  \centering
  \caption{Axioms and $\mu^k_{\text{DP}}$}
  \label{TBL_HDP}
  \begin{tabular}{l|cccc}
    \hline
    $k$ & WSYM & DUP & MWC & AVE \\
    \hline
    1 & \checkmark & - & \checkmark & \checkmark \\
    2 (Prop. \ref{PROP_HDP}) & \checkmark & \checkmark & \checkmark & \checkmark \\
    3 and more & \checkmark & - & \checkmark & - \\
    \hline
  \end{tabular}
\end{table}

\subsection{Extension of the Jaccard index and the Dice index}

As introduced in Section \ref{SEC_MEASURE}, the Dice index is defined as
\[
\delta(S_1,...,S_k)=k\cdot \frac{ |S_1\cap ... \cap S_k|}{|S_1|+...+|S_k|}.
\]
Since the Dice index directly measures the extent of overlap among sets, it should be natural to ask whether it can itself serve as a measure of power concentration.
In particular, we examine which of the axioms introduced above are satisfied if such an overlap index is applied directly to the collection of minimal winning coalitions.
Formally, we define the following concentration measure based on the Dice index:
\[
\mu_D(N,W)=|M(W)| \cdot \frac{ |\cap_{S\in M(W)} S|}{\sum_{S\in M(W)}|S|}.
\]

Moreover, as another variation, one may also consider a concentration measure derived from the Jaccard index (Jaccard, 1901), which is defined as the ratio of the size of the intersection to that of the union of sets: formally,
\[
\mu_J(N,W)=\frac{ |\cap_{S\in M(W)} S|}{|\cup_{S\in M(W)} S|}.
\]

Table \ref{TBL_DJ} summarizes which axioms are satisfied by these two measures.
\begin{table}[htbp]
  \centering
  \caption{Axioms, $\mu_D$, and $\mu_J$}
  \label{TBL_DJ}
  \begin{tabular}{c|cccc}
    \hline
     & WSYM & DUP & MWC & AVE \\
    \hline
    $\mu_D$ & - & - & \checkmark & - \\
    $\mu_J$ & - & - & \checkmark & - \\
    \hline
  \end{tabular}
\end{table}
A common drawback of these concentration measures is that they assign zero to any voting game whose minimal winning coalitions have an empty intersection.
Note that the proposed measure $\sum_{i\in N} DP_i(N,W)^2$ avoids this drawback.
This is because, as shown in the proof, it can be equivalently expressed as $\frac{1}{|M(W)|^2} \sum_{(S,T)\in M(W)^2} \frac{ |S\cap T|}{|S|\cdot|T|}$.
Since the summation ranges over all pairs of minimal winning coalitions, it includes in particular the cases with $S=T$. Therefore, even when the intersection of all minimal winning coalitions is empty, the measure avoids being zero.

\section{Application to UNSC}\label{SEC_APL}
In this section, we apply the proposed measure to the UNSC and its variants. As discussed in the introduction, the UNSC consists of five permanent members and ten non-permanent members. A decision is made if it receives affirmative votes from all five permanent members and at least four non-permanent members.

We formalize this decision rule as follows. Let $n_1$ ($n_2$) denote the number of permanent (non-permanent) members, and let $k_1$ ($k_2$) denote the number of affirmative votes required from permanent (non-permanent) members in order for a proposal to be adopted.
Under the current UNSC rule, these parameters are given by $(n_1,k_1; n_2, k_2)=(5,5; 10,4)$.
Let $DP_1(n_1,k_1; n_2, k_2)$ denote the DP index of a permanent member, and let $DP_2(n_1,k_1; n_2, k_2)$ represent the DP index of a non-permanent member.
In general, for every parameter combination $(n_1,k_1; n_2, k_2)$ with $1\leq k_1\leq n_1$ and $1\leq k_2\leq n_2$, the DP index takes the following form:
\begin{itemize}
\item $DP_1(n_1,k_1; n_2, k_2) = \frac{1}{\binom{n_1}{k_1} \binom{n_2}{k_2}} \binom{n_1-1}{k_1-1}\binom{n_2}{k_2} \frac{1}{k_1+k_2}=\frac{k_1}{n_1} \frac{1}{k_1+k_2}$,
\item $DP_2(n_1,k_1; n_2, k_2)  =\frac{1}{\binom{n_1}{k_1} \binom{n_2}{k_2}} \binom{n_1}{k_1}\binom{n_2-1}{k_2-1} \frac{1}{k_1+k_2}=\frac{k_2}{n_2} \frac{1}{k_1+k_2}$.
\end{itemize}

These expressions reveal that DP voting power is proportional to each group's relative approval requirement $\frac{k_{\ell}}{n_{\ell}}$ and is scaled by the total size of minimal winning coalitions $k_1+k_2$.
Applying the concentration measure, the power concentration of $(n_1,k_1; n_2, k_2)$ is given by
\[
\mu(n_1,k_1; n_2, k_2) = \frac{1}{(k_1+k_2)^2} \left( \frac{k_1^2}{n_1} + \frac{k_2^2}{n_2} \right).
\]
Substituting the UNSC parameters $(5,5; 10,4)$ into the formula yields
\[
\mu^* = \frac{1}{81} \left( \frac{25}{5} + \frac{16}{10} \right) = \frac{11}{135} =0.08148... .
\]
This value serves as a benchmark for comparing the current UNSC setting to alternative parameters.

Table \ref{TBL_UNSC} shows the values of $\mu(n_1,k_1; n_2,k_2)$ obtained by fixing $n_1=5$ and $n_2=10$ while varying the approval thresholds $k_1$ and $k_2$. For example, the value at $(k_1,k_2)=(5,4)$ refers to the benchmark value $\mu^*$.
In the table, higher values of the power concentration are represented by darker shades.

\begin{table}[htbp]
\centering
\caption{Concentration levels for $(n_1,n_2)=(5,10)$}\label{TBL_UNSC}
{\small
\setlength{\tabcolsep}{3pt} 
\begin{tabular}{c|cccccccccc}
\hline
$k_1 \backslash k_2$
& 1 & 2 & 3 & 4 & 5 & 6 & 7 & 8 & 9 & 10 \\
\hline
1
& \cellcolor{gray!17}0.07500
& 0.06667
& \cellcolor{gray!9}0.06875
& \cellcolor{gray!14}0.07200
& \cellcolor{gray!17}0.07500
& \cellcolor{gray!20}0.07755
& \cellcolor{gray!22}0.07969
& \cellcolor{gray!24}0.08148
& \cellcolor{gray!26}0.08300
& \cellcolor{gray!28}0.08430 \\
2
& \cellcolor{gray!28}0.10000
& \cellcolor{gray!17}0.07500
& \cellcolor{gray!10}0.06800
& 0.06667
& \cellcolor{gray!4}0.06735
& \cellcolor{gray!9}0.06875
& \cellcolor{gray!12}0.07037
& \cellcolor{gray!14}0.07200
& \cellcolor{gray!16}0.07355
& \cellcolor{gray!17}0.07500 \\
3
& \cellcolor{gray!34}0.11875
& \cellcolor{gray!25}0.08800
& \cellcolor{gray!17}0.07500
& \cellcolor{gray!11}0.06939
& \cellcolor{gray!5}0.06719
& 0.06667
& \cellcolor{gray!3}0.06700
& \cellcolor{gray!5}0.06777
& \cellcolor{gray!9}0.06875
& \cellcolor{gray!11}0.06982 \\
4
& \cellcolor{gray!38}0.13200
& \cellcolor{gray!28}0.10000
& \cellcolor{gray!23}0.08367
& \cellcolor{gray!17}0.07500
& \cellcolor{gray!12}0.07037
& \cellcolor{gray!10}0.06800
& \cellcolor{gray!4}0.06694
& 0.06667
& \cellcolor{gray!4}0.06686
& \cellcolor{gray!4}0.06735 \\
5
& \cellcolor{gray!40}0.14167
& \cellcolor{gray!33}0.11020
& \cellcolor{gray!26}0.09219
& \cellcolor{gray!19}\textbf{0.08148}
& \cellcolor{gray!17}0.07500
& \cellcolor{gray!13}0.07107
& \cellcolor{gray!9}0.06875
& \cellcolor{gray!5}0.06746
& \cellcolor{gray!2}0.06684
& 0.06667 \\
\hline
\end{tabular}
}
\end{table}

For fixed group sizes $n_1$ and $n_2$, power concentration is minimized when the relative approval requirements are equal across groups, that is, when $\frac{k_1}{n_1} = \frac{k_2}{n_2}$.
In the present case, where $n_1=5$ and $n_2=10$, this condition is satisfied when $k_2 = 2k_1$. Consequently, $\mu(n_1,k_1; n_2,k_2)$ attains its minimum at the parameter combinations $(k_1,k_2) = (1,2), \ldots, (5,10)$.

Table \ref{TBL_UNSC} also shows that the concentration measure is not monotonic in $k_1$ and $k_2$. Consider the column corresponding to $k_2=4$. As $k_1$ decreases from $5$ to $2$, the required level of consent among the permanent members falls, and power concentration declines accordingly. However, when $k_1$ decreases further from $2$ to $1$, the level of concentration increases. This is because non-permanent members become relatively more powerful than permanent members at this point.

Table \ref{TBL_n1n2} reports the values of $\mu(n_1,k_1; n_2,k_2)$ obtained by fixing $k_1=5$ and $k_2=4$, while varying $n_1=5,\ldots,11$ and $n_2=4,\ldots,14$. A comparison of Tables \ref{TBL_UNSC} and \ref{TBL_n1n2} yields the following observations.
\begin{itemize}
\item Reducing the approval requirement for permanent members from $k_1=5$ to $k_1=4$ leads to a smaller decrease in power concentration than increasing the number of permanent members from $n_1=5$ to $n_1=6$ while keeping $k_1=5$; 
$0.08148 - 0.07500 < 0.08148 - 0.07119$.

\item By contrast, for non-permanent members, increasing the approval requirement from $k_2=4$ to $k_2=5$ leads to a larger reduction in power concentration than increasing the number of non-permanent members from $n_2=10$ to $n_2=11$ while keeping $k_2=4$; $0.08148 - 0.07500 > 0.08148 - 0.07969$.
\end{itemize}

\begin{table}[htbp]
\centering
\caption{Concentration levels for $(k_1,k_2)=(5,4)$}
\label{TBL_n1n2}
{\small
\setlength{\tabcolsep}{3pt}
\begin{tabular}{c|ccccccccccc}
\hline
$n_1 \backslash n_2$
& 4 & 5 & 6 & 7 & 8 & 9 & 10 & 11 & 12 & 13 & 14 \\
\hline
5
& \cellcolor{gray!40}0.11111
& \cellcolor{gray!34}0.10123
& \cellcolor{gray!30}0.09465
& \cellcolor{gray!26}0.08995
& \cellcolor{gray!23}0.08642
& \cellcolor{gray!21}0.08368
& \cellcolor{gray!19}\textbf{0.08148}
& \cellcolor{gray!17}0.07969
& \cellcolor{gray!15}0.07819
& \cellcolor{gray!14}0.07692
& \cellcolor{gray!13}0.07584 \\
6
& \cellcolor{gray!33}0.10082
& \cellcolor{gray!28}0.09095
& \cellcolor{gray!24}0.08436
& \cellcolor{gray!21}0.07966
& \cellcolor{gray!18}0.07613
& \cellcolor{gray!16}0.07339
& \cellcolor{gray!14}0.07119
& \cellcolor{gray!12}0.06940
& \cellcolor{gray!10}0.06790
& \cellcolor{gray!9}0.06664
& \cellcolor{gray!8}0.06555 \\
7
& \cellcolor{gray!30}0.09347
& \cellcolor{gray!25}0.08360
& \cellcolor{gray!22}0.07701
& \cellcolor{gray!19}0.07231
& \cellcolor{gray!16}0.06878
& \cellcolor{gray!14}0.06604
& \cellcolor{gray!12}0.06384
& \cellcolor{gray!10}0.06205
& \cellcolor{gray!9}0.06055
& \cellcolor{gray!8}0.05929
& \cellcolor{gray!7}0.05820 \\
8
& \cellcolor{gray!27}0.08796
& \cellcolor{gray!23}0.07809
& \cellcolor{gray!20}0.07150
& \cellcolor{gray!17}0.06680
& \cellcolor{gray!15}0.06327
& \cellcolor{gray!12}0.06053
& \cellcolor{gray!11}0.05833
& \cellcolor{gray!9}0.05654
& \cellcolor{gray!8}0.05504
& \cellcolor{gray!7}0.05377
& \cellcolor{gray!6}0.05269 \\
9
& \cellcolor{gray!24}0.08368
& \cellcolor{gray!20}0.07380
& \cellcolor{gray!17}0.06722
& \cellcolor{gray!14}0.06251
& \cellcolor{gray!12}0.05898
& \cellcolor{gray!10}0.05624
& \cellcolor{gray!9}0.05405
& \cellcolor{gray!7}0.05225
& \cellcolor{gray!6}0.05075
& \cellcolor{gray!5}0.04949
& \cellcolor{gray!4}0.04840 \\
10
& \cellcolor{gray!22}0.08025
& \cellcolor{gray!18}0.07037
& \cellcolor{gray!15}0.06379
& \cellcolor{gray!12}0.05908
& \cellcolor{gray!10}0.05556
& \cellcolor{gray!8}0.05281
& \cellcolor{gray!7}0.05062
& \cellcolor{gray!5}0.04882
& \cellcolor{gray!4}0.04733
& \cellcolor{gray!3}0.04606
& \cellcolor{gray!2}0.04497 \\
11
& \cellcolor{gray!20}0.07744
& \cellcolor{gray!16}0.06756
& \cellcolor{gray!13}0.06098
& \cellcolor{gray!10}0.05628
& \cellcolor{gray!8}0.05275
& \cellcolor{gray!6}0.05001
& \cellcolor{gray!5}0.04781
& \cellcolor{gray!4}0.04602
& \cellcolor{gray!3}0.04452
& \cellcolor{gray!2}0.04325
& \cellcolor{gray!1}0.04217 \\
\hline
\end{tabular}
}
\end{table}

Moreover, in 2005, Brazil, Germany, India, and Japan jointly submitted a proposal to reform the approval rule of the UNSC. The proposal called for an expansion of the council by increasing the number of permanent members from $n_1 = 5$ to $11$ and the number of non‑permanent members from $n_2 = 10$ to $14$. The proposal was ultimately rejected.\footnote{An amendment to the UNSC approval rule requires the approval of at least two thirds of the members of the United Nations General Assembly. In addition, ratification requires the approval of at least two thirds of the UN members, including all permanent members of the Security Council.}
We, however, revisit this reform proposal from the perspective of power concentration by comparing it with the current rule. 
Since the proposal did not explicitly specify the approval thresholds $(k_1,k_2)$ for permanent and non-permanent members, we compute the level of concentration for all possible threshold pairs $k_1 \in \{1,\ldots,n_1\}$ and $k_2 \in \{1,\ldots,n_2\}$. 
Table \ref{TBL_k1k2_11_14} shows the resulting concentration levels for the proposed case $(n_1,n_2)=(11,14)$.

If the current thresholds $(k_1,k_2)=(5,4)$ are maintained under the proposed expansion, the resulting level of concentration is approximately $0.042$, which is about half of the level observed under the current rule. 
Moreover, even in the highest case, corresponding to $(k_1,k_2)=(11,1)$, the level of concentration remains lower than that of the current rule. 
This result arises because expanding the council by a total of ten additional members leads to a more diffuse distribution of power across members.

\begin{table}[htbp]
\centering
\caption{Concentration levels for $(n_1,n_2)=(11,14)$}
\label{TBL_k1k2_11_14}
{\small
\setlength{\tabcolsep}{3pt}
\begin{tabular}{c|cccccccccccccc}
\hline
$k_1 \backslash k_2$
& 1 & 2 & 3 & 4 & 5 & 6 & 7 & 8 & 9 & 10 & 11 & 12 & 13 & 14 \\
\hline
1  & \cellcolor{gray!1}0.041 & \cellcolor{gray!2}0.042 & \cellcolor{gray!6}0.046 & \cellcolor{gray!10}0.049
   & \cellcolor{gray!13}0.052 & \cellcolor{gray!15}0.054 & \cellcolor{gray!17}0.056 & \cellcolor{gray!19}0.058
   & \cellcolor{gray!20}0.059 & \cellcolor{gray!22}0.060 & \cellcolor{gray!23}0.061 & \cellcolor{gray!23}0.061
   & \cellcolor{gray!24}0.062 & \cellcolor{gray!26}0.063 \\
2  & \cellcolor{gray!9}0.048 & \cellcolor{gray!1}0.041 & \cellcolor{gray!0}0.040 & \cellcolor{gray!2}0.042
   & \cellcolor{gray!4}0.044 & \cellcolor{gray!6}0.046 & \cellcolor{gray!9}0.048 & \cellcolor{gray!10}0.049
   & \cellcolor{gray!12}0.051 & \cellcolor{gray!13}0.052 & \cellcolor{gray!15}0.053 & \cellcolor{gray!15}0.054
   & \cellcolor{gray!17}0.055 & \cellcolor{gray!19}0.056 \\
3  & \cellcolor{gray!17}0.056 & \cellcolor{gray!4}0.044 & \cellcolor{gray!1}0.041 & \cellcolor{gray!0}0.040
   & \cellcolor{gray!1}0.041 & \cellcolor{gray!2}0.042 & \cellcolor{gray!3}0.043 & \cellcolor{gray!6}0.045
   & \cellcolor{gray!7}0.046 & \cellcolor{gray!9}0.047 & \cellcolor{gray!10}0.048 & \cellcolor{gray!12}0.049
   & \cellcolor{gray!13}0.050 & \cellcolor{gray!15}0.051 \\
4  & \cellcolor{gray!23}0.061 & \cellcolor{gray!9}0.048 & \cellcolor{gray!3}0.043 & \cellcolor{gray!1}0.041
   & \cellcolor{gray!0}0.040 & \cellcolor{gray!0}0.040 & \cellcolor{gray!1}0.041 & \cellcolor{gray!2}0.042
   & \cellcolor{gray!3}0.043 & \cellcolor{gray!4}0.044 & \cellcolor{gray!6}0.045 & \cellcolor{gray!7}0.046
   & \cellcolor{gray!9}0.047 & \cellcolor{gray!10}0.048 \\
5  & \cellcolor{gray!28}0.065 & \cellcolor{gray!13}0.052 & \cellcolor{gray!6}0.046 & \cellcolor{gray!2}\textbf{0.042}
   & \cellcolor{gray!1}0.041 & \cellcolor{gray!0}0.040 & \cellcolor{gray!0}0.040 & \cellcolor{gray!0}0.040
   & \cellcolor{gray!1}0.041 & \cellcolor{gray!2}0.042 & \cellcolor{gray!3}0.043 & \cellcolor{gray!3}0.043
   & \cellcolor{gray!4}0.044 & \cellcolor{gray!6}0.045 \\
6  & \cellcolor{gray!31}0.068 & \cellcolor{gray!17}0.056 & \cellcolor{gray!9}0.048 & \cellcolor{gray!4}0.044
   & \cellcolor{gray!2}0.042 & \cellcolor{gray!1}0.041 & \cellcolor{gray!0}0.040 & \cellcolor{gray!0}0.040
   & \cellcolor{gray!0}0.040 & \cellcolor{gray!1}0.041 & \cellcolor{gray!1}0.041 & \cellcolor{gray!2}0.042
   & \cellcolor{gray!3}0.043 & \cellcolor{gray!3}0.043 \\
7  & \cellcolor{gray!34}0.071 & \cellcolor{gray!20}0.059 & \cellcolor{gray!12}0.051 & \cellcolor{gray!6}0.046
   & \cellcolor{gray!3}0.043 & \cellcolor{gray!2}0.042 & \cellcolor{gray!1}0.041 & \cellcolor{gray!0}0.040
   & \cellcolor{gray!0}0.040 & \cellcolor{gray!0}0.040 & \cellcolor{gray!0}0.040 & \cellcolor{gray!1}0.041
   & \cellcolor{gray!1}0.041 & \cellcolor{gray!2}0.042 \\
8  & \cellcolor{gray!37}0.073 & \cellcolor{gray!23}0.061 & \cellcolor{gray!14}0.053 & \cellcolor{gray!9}0.048
   & \cellcolor{gray!6}0.045 & \cellcolor{gray!3}0.043 & \cellcolor{gray!1}0.041 & \cellcolor{gray!1}0.041
   & \cellcolor{gray!0}0.040 & \cellcolor{gray!0}0.040 & \cellcolor{gray!0}0.040 & \cellcolor{gray!0}0.040
   & \cellcolor{gray!1}0.041 & \cellcolor{gray!1}0.041 \\
9  & \cellcolor{gray!38}0.074 & \cellcolor{gray!27}0.063 & \cellcolor{gray!17}0.056 & \cellcolor{gray!11}0.050
   & \cellcolor{gray!7}0.047 & \cellcolor{gray!4}0.044 & \cellcolor{gray!2}0.042 & \cellcolor{gray!1}0.041
   & \cellcolor{gray!1}0.041 & \cellcolor{gray!0}0.040 & \cellcolor{gray!0}0.040 & \cellcolor{gray!0}0.040
   & \cellcolor{gray!0}0.040 & \cellcolor{gray!0}0.040 \\
10 & \cellcolor{gray!40}0.076 & \cellcolor{gray!29}0.065 & \cellcolor{gray!19}0.058 & \cellcolor{gray!13}0.052
   & \cellcolor{gray!9}0.048 & \cellcolor{gray!6}0.046 & \cellcolor{gray!4}0.044 & \cellcolor{gray!2}0.042
   & \cellcolor{gray!1}0.041 & \cellcolor{gray!1}0.041 & \cellcolor{gray!0}0.040 & \cellcolor{gray!0}0.040
   & \cellcolor{gray!0}0.040 & \cellcolor{gray!0}0.040 \\
11 & \cellcolor{gray!40}0.077 & \cellcolor{gray!31}0.067 & \cellcolor{gray!20}0.059 & \cellcolor{gray!15}0.054
   & \cellcolor{gray!11}0.050 & \cellcolor{gray!7}0.047 & \cellcolor{gray!5}0.045 & \cellcolor{gray!3}0.043
   & \cellcolor{gray!2}0.042 & \cellcolor{gray!1}0.041 & \cellcolor{gray!1}0.041 & \cellcolor{gray!0}0.040
   & \cellcolor{gray!0}0.040 & \cellcolor{gray!0}0.040 \\
\hline
\end{tabular}
}
\end{table}

\section{Conclusion}\label{SEC_CONC}
This paper has proposed a measure of power concentration in monotonic voting games. The measure is characterized as the unique function that satisfies weak symmetry, the replication property, consistency with respect to the set of minimal winning coalitions, and an averaging property for concentration.

The proposed measure has two interpretations. First, it can be understood as the sum of squared Deegan-Packel power indices. From this perspective, the level of concentration is interpreted as the aggregation of the Deegan-Packel power index by the Herfindahl-Hirschman index. 
Second, the measure can also be interpreted as capturing the degree of overlap among winning coalitions. Since any monotonic voting game can be formulated in terms of its set of minimal winning coalitions, the extent to which these coalitions overlap provides a natural evaluation of power concentration.

This study has focused on measuring power concentration, whereas we believe that a certain kind of ``decisiveness'' of a voting game is another important issue. For example, a voting game in which a proposal is approved only if all nine players vote in favor (all-player unanimity game) and a voting game in which approval requires five out of nine votes (simple majority voting game) are both symmetric and therefore exhibit the same level of concentration.
However, the ease with which proposals can be approved differs substantially between them. This difference should be associated with the ``abundance'' of winning coalitions. More specifically, a voting game with many small winning coalitions can be expected to be more permissive in approving proposals. Establishing a formal measure of decisiveness remains an important topic for future research.

\section*{Declaration}

\noindent\textbf{Funding.} The author gratefully acknowledges the financial support from JSPS: No.22H00829.

\noindent\textbf{Conflict of interest.} The author declares that there are no conflicts of interest.

\noindent\textbf{Data availability.} This work is not based on any empirical or simulated data.

\end{document}